\documentclass[3p,preprint]{elsarticle}
\usepackage{amssymb}
\usepackage{amsmath}
\usepackage{amsthm}
\usepackage{color}
\usepackage{enumitem}
\usepackage{mathtools}
\usepackage{enumitem}
\theoremstyle{definition}
\newtheorem{definition}{Definition}
\newtheorem{proposition}{Proposition}
\newtheorem{observation}{Observation}
\newtheorem{lemma}{Lemma}
\newtheorem{claim}{Claim}
\newtheorem{example}{Example}

\newtheorem{corollary}{Corollary}

\newtheorem{theorem}{Theorem}

\usepackage{enumitem}
\setlist[itemize]{itemsep=2pt,parsep=0pt,partopsep=0pt}
\setlist[enumerate]{itemsep=2pt,parsep=0pt,partopsep=0pt}

\usepackage{caption}
\usepackage{tikz}
\usetikzlibrary{arrows.meta,positioning,fit}
\usepackage{subcaption}
\usepackage{booktabs}

\begin{document}
\usetikzlibrary{arrows.meta}
    \begin{frontmatter}
        \title{Condorcet-type properties of the linear ordering problem with ties}
         \author[label1]{Daichi Kawashima}
         \affiliation[label1]{organization={Department of Advanced Sciences, Hosei University},
                addressline={3-7-2}, 
                city={Koganei},
                postcode={184-8584}, 
                state={Tokyo Prefecture},
                country={Japan}}
        \author[label1]{Noriyoshi Sukegawa}
    
\begin{abstract}
The Kemeny rule aggregates multiple strict rankings into a single strict ranking that minimizes the sum of its distances from the input rankings. 
The resulting optimization problem, called the Kemeny problem (\texttt{KP}), is a special case of the linear ordering problem (\texttt{LOP}). 
The Kemeny rule satisfies several desirable properties in social choice theory, including the extended Condorcet criterion (\texttt{XCC}). 
Ando et al.\ strengthened this result by introducing the strong Condorcet criterion (\texttt{SCC}) and showing that it holds for every optimal solution to an arbitrary \texttt{LOP} instance. 
Yoo and Escobedo extended the Kemeny rule to rankings with ties and showed that the resulting rule satisfies the non-strict extended Condorcet criterion (\texttt{NXCC}). 
This criterion gives a condition under which one candidate must be ranked strictly above another in every optimal solution. 
In this paper, we introduce the non-strict strong Condorcet criterion (\texttt{NSCC}), a counterpart of the \texttt{SCC} for rankings with ties, and show that it holds for every optimal solution to an arbitrary instance of the linear ordering problem with ties (\texttt{LOPT}). 
We also establish a complementary structural property that gives conditions under which two candidates must be tied in every optimal solution to an arbitrary \texttt{LOPT} instance.
\end{abstract}


\end{frontmatter}
    
\section{Introduction}

Rank aggregation has numerous applications, including meta-search~\cite{Dwork2001} and biological databases~\cite{Sese2001}, and also plays a central role in social choice theory.
Among rank-aggregation rules, the Kemeny rule (or Kemeny--Young rule), proposed by Kemeny~\cite{Kemeny1959}, is well known for satisfying important criteria in social choice theory, including consistency and the Condorcet criterion (\texttt{CC}). 
The simultaneous satisfaction of these two criteria is known to be rare~\cite{Young1978}. 
The objective of this study is to clarify the structural properties of the outputs of the Kemeny rule when ties are allowed in rankings. 

\subsection{Background}

Suppose that there are multiple voters and each voter provides a ranking of a common set of candidates. 
A Kemeny ranking is a ranking that minimizes the sum of its distances from the rankings provided by the voters. 
The Kemeny rule is a rule that always returns a Kemeny ranking.  
The problem of finding a Kemeny ranking is called the Kemeny problem (\texttt{KP}), and is known to be a special case of the linear ordering problem (\texttt{LOP}). 
In the \texttt{KP}, the pairwise weights are derived from the input rankings, whereas the \texttt{LOP} allows arbitrary pairwise weights and seeks a ranking of maximum total weight. 
Therefore, properties of optimal solutions to the \texttt{LOP} can be used to study the Kemeny rule.

A candidate who defeats every other candidate in pairwise comparisons is called a Condorcet winner.
An aggregation rule satisfies the Condorcet criterion (\texttt{CC}) if the rule ranks the Condorcet winner first whenever one exists. 
Pairwise-comparison outcomes can be represented by a simple directed graph, called the majority graph. 
Truchon~\cite{Truchon1998} used this majority graph to introduce the extended Condorcet criterion (\texttt{XCC}) and showed that the Kemeny rule satisfies it. 
Ando et al.~\cite{Andoetal2022} strengthened the \texttt{XCC} to the \texttt{SCC} and showed that the \texttt{SCC} holds not only for Kemeny rankings but for every optimal solution to an arbitrary \texttt{LOP} instance.
The \texttt{SCC} implies the \texttt{XCC}, which in turn implies the \texttt{CC}.

The studies described above assume strict rankings.
Yoo and Escobedo~\cite{YooEscobedo2021} extended the Kemeny rule to rankings with ties and introduced the non-strict extended Condorcet criterion (\texttt{NXCC}), a counterpart of the \texttt{XCC} for this setting.
The \texttt{NXCC} gives conditions under which one candidate must be ranked strictly above another in every optimal solution. 
However, it neither provides a counterpart of the stronger \texttt{SCC} nor gives conditions under which two candidates must be tied in every optimal solution.

Just as the \texttt{KP} is a special case of the \texttt{LOP}, the Kemeny problem with ties (\texttt{KPT}) can be generalized by allowing arbitrary pairwise weights. 
We call the resulting optimization problem the linear ordering problem with ties (\texttt{LOPT}). 
Figure~\ref{fig:problem-classes} summarizes the relations among these four problem classes. 
The horizontal arrows allow arbitrary pairwise weights, while the vertical arrows allow ties.

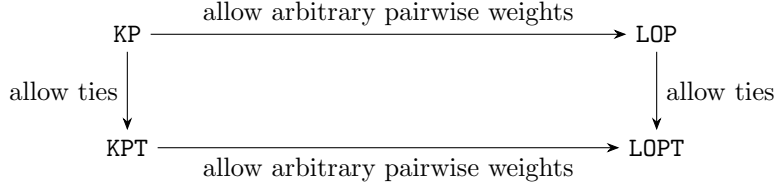
\begin{figure}[t]
\centering
\begin{tikzpicture}[>={Stealth}]
\node (KP)   at (0,0)   {\texttt{KP}};
\node (LOP)  at (7,0)   {\texttt{LOP}};
\node (KPT)  at (0,-1.5)  {\texttt{KPT}};
\node (LOPT) at (7,-1.5)  {\texttt{LOPT}};
\draw[->] (KP) -- node[above] {allow arbitrary pairwise weights} (LOP);
\draw[->] (KP) -- node[left] {allow ties} (KPT);
\draw[->] (LOP) -- node[right] {allow ties} (LOPT);
\draw[->] (KPT) -- node[below] {allow arbitrary pairwise weights} (LOPT);
\end{tikzpicture}
\caption{Relations among the four problem classes.}
\label{fig:problem-classes}
\end{figure}

\subsection{Our contributions}

We investigate the structural properties of optimal solutions to the \texttt{LOPT} and establish the following two results.

\begin{enumerate}
\item
We introduce the non-strict strong Condorcet criterion (\texttt{NSCC}), a counterpart of the \texttt{SCC} for rankings with ties, and show that it holds for every optimal solution to an arbitrary \texttt{LOPT} instance. 
The \texttt{NSCC} implies the \texttt{NXCC}, providing a finer characterization of strict comparisons in optimal rankings.

\item
We establish a complementary structural property that gives conditions under which two candidates must be tied in every optimal solution to an arbitrary \texttt{LOPT} instance. 
This property concerns ties and is not captured by either the \texttt{NXCC} or the \texttt{NSCC}.
\end{enumerate}

Our analysis extends the graph-theoretic approach of Ando et al.~\cite{Andoetal2022} from the \texttt{LOP} to the \texttt{LOPT}. 
To accommodate ties, we represent rankings with ties by unicycle-free directed graphs and establish an equivalence between the \texttt{LOPT} and a corresponding graph optimization problem.

\subsection{Organization}

The remainder of this paper is organized as follows. 
Section~\ref{Preliminaries} introduces the necessary preliminaries and formulates the \texttt{LOPT}. 
Section~\ref{MainResults} states the main results and gives an example.
Section~\ref{GraphReformulation} formulates the \texttt{LOPT} as a graph optimization problem.
Section~\ref{ProofsMainResults} proves the main results.
Section~\ref{RelatedWork} discusses related work.
Section~\ref{Conclusion} gives concluding remarks.

\section{Preliminaries}\label{Preliminaries}

\subsection{Kemeny rule and the linear ordering problem}

Let $V$ be a set of $n$ candidates. 
Let $\mathcal{U}$ denote the set of all unordered pairs of distinct
candidates: 
\begin{align*}
\mathcal{U}\coloneqq \{\{u,v\}:u,v\in V,\ u\neq v\}.
\end{align*} 
A strict ranking is a bijection $\pi:V\to[n]$, where $[n]\coloneqq\{1,2,\ldots,n\}$. 
If $\pi(u)<\pi(v)$, we say that $u$ is ranked above $v$ and write $u\succ_\pi v$.
For any two strict rankings $\pi,\rho$, their Kendall tau distance is
defined as
\begin{align*}
d_{\mathrm{K}}(\pi,\rho)
\coloneqq
\frac{1}{2}
\sum_{\{u,v\}\in\mathcal{U}}
\left| \operatorname{sgn}(\pi(u)-\pi(v)) - \operatorname{sgn}(\rho(u)-\rho(v)) \right|,
\end{align*}
where $\operatorname{sgn}(x)$ takes the values $-1$, $0$, and $1$ when $x<0$, $x=0$, and $x>0$, respectively.
Since $\pi$ and $\rho$ are strict, each sign function takes a value in $\{-1,1\}$.
Hence, $d_{\mathrm{K}}(\pi,\rho)$ counts the candidate pairs whose relative orders differ between $\pi$ and $\rho$.

Suppose that $\Pi$ is a finite collection of strict rankings. 
A strict ranking $\pi$ that minimizes
\begin{align*}
f^{\texttt{KP}}_\Pi(\pi)
\coloneqq
\sum_{\rho\in\Pi} d_{\mathrm{K}}(\pi,\rho)
\end{align*}
over all strict rankings on $V$ is called a Kemeny ranking. 
The problem of finding a Kemeny ranking is called the Kemeny problem (\texttt{KP}), which is known to be NP-hard even for $|\Pi|=4$~\cite{Bartholdi1989}.

The linear ordering problem (\texttt{LOP}) generalizes the \texttt{KP}. 
For each $\{u,v\}\in \mathcal{U}$, let $w^\succ_{uv}$ and $w^\succ_{vu}$ denote the nonnegative weights assigned when $u$ is ranked above $v$ and when $v$ is ranked above $u$, respectively. 
Let $W$ denote the collection of all these weights. 
The \texttt{LOP} seeks a ranking $\pi$ that maximizes
\begin{align*}
f^{\texttt{LOP}}_W(\pi)
\coloneqq
\sum_{\{u,v\}\in \mathcal{U}}
\left( 
w^\succ_{uv} \mathbf{1}[u\succ_\pi v]
+
w^\succ_{vu} \mathbf{1}[v\succ_\pi u]
\right),
\end{align*}
where $\mathbf{1}[\cdot]$ denotes the indicator function, which equals one if the statement in brackets holds and zero otherwise. 
We denote this instance as $\texttt{LOP}(W)$.
The \texttt{KP} reduces to the \texttt{LOP} by setting 
\begin{align*}
w^\succ_{uv}
\coloneqq
\left|\{\rho \in \Pi : u\succ_{\rho} v\}\right|.
\end{align*}
Indeed, for every ranking $\pi$, we have 
\begin{align*}
f^{\texttt{KP}}_\Pi(\pi)+f^{\texttt{LOP}}_W(\pi)
=
|\Pi||\mathcal U|.
\end{align*}
Consequently, a ranking minimizes the \texttt{KP} objective if and only if it maximizes that of the corresponding \texttt{LOP} instance.

\subsection{Theoretical properties of optimal solutions of the linear ordering problem}
We next describe three Condorcet-type properties of optimal solutions to the \texttt{LOP}. 
\begin{definition}[Condorcet winner]
For any instance $\texttt{LOP}(W)$, a candidate $u$ is a Condorcet winner if
\begin{align*}
w^\succ_{uv}>w^\succ_{vu}
\quad
\text{for every }v\in V\setminus\{u\}.
\end{align*}
\end{definition}
The Condorcet criterion requires such a candidate to be ranked first. 
For the \texttt{LOP}, this requirement can be stated as the following property of its optimal solutions.
\begin{proposition}[\texttt{CC}]
For any instance $\texttt{LOP}(W)$, if $u$ is a Condorcet winner of $\texttt{LOP}(W)$, then $\pi(u)=1$ for every optimal solution $\pi$.
\end{proposition}

A Condorcet winner does not necessarily exist because pairwise preferences may form cycles. 
To describe the structure of optimal rankings in such cases, we represent the pairwise weights by directed graphs. 
In what follows, we regard each candidate as a vertex.

\begin{definition}[Majority graphs for the \texttt{LOP}]
For any instance $\texttt{LOP}(W)$, define the majority graph $G_W=(V,A_W)$ and the strict majority graph $G_W^+=(V,A_W^+)$ by
\begin{align*}
A_W
&=
\{(u,v):\{u,v\}\in\mathcal U,\
w^\succ_{uv}\geq w^\succ_{vu}\},\\
A_W^+
&=
\{(u,v):\{u,v\}\in\mathcal U,\
w^\succ_{uv}>w^\succ_{vu}\}.
\end{align*}
Let $\mathcal{P}_W$ and $\mathcal{P}_W^+$ denote their strongly connected component decompositions.
\end{definition}

A directed graph is said to be semicomplete if at least one of $(u,v)$ and $(v,u)$ is present for every pair of distinct vertices $u,v$. 
By definition, the majority graph $G_W$ is semicomplete.
Between any two distinct strongly connected components of $G_W$, all arcs have the same direction and correspond to strict pairwise comparisons. 
We therefore write
$\mathcal{P}_W=\{V_1,V_2,\ldots,V_m\}$
and index the components so that
\begin{align*}
V_i\times V_j\subseteq A_W^+
\quad
\text{for every }1\le i<j\le m.
\end{align*}
In other words, for any $V_i,V_j\in\mathcal{P}_W$ with $i<j$, every candidate in $V_i$ strictly defeats every candidate in $V_j$ in their pairwise comparison.

\begin{proposition}[\texttt{XCC}]
For any instance $\texttt{LOP}(W)$, let $V_i,V_j\in\mathcal{P}_W$ with $i<j$. 
Then, for every $u\in V_i$, every $v\in V_j$, and every optimal solution $\pi$ to $\texttt{LOP}(W)$, we have $\pi(u)<\pi(v)$.
\end{proposition}

If $w^\succ_{uv}=w^\succ_{vu}$, then $G_W$ contains the two arcs $(u,v)$ and $(v,u)$, whereas $G_W^+$ contains neither. 
Therefore, the graph $G_W^+$ may yield a finer strongly connected component decomposition than $G_W$. 
The \texttt{SCC} is based on this finer decomposition.

\begin{proposition}[\texttt{SCC}]\label{scc}
For any instance $\texttt{LOP}(W)$, let $u,v\in V$ be such that they belong to distinct components of $\mathcal{P}_W^+$ and $(u,v)\in A_W^+$. 
Then, for every optimal solution $\pi$ to $\texttt{LOP}(W)$, we have $\pi(u)<\pi(v)$.
\end{proposition}

These criteria form a hierarchy. 
Since $\mathcal{P}_W^+$ refines $\mathcal{P}_W$, the \texttt{SCC} implies the \texttt{XCC}. 
Moreover, a Condorcet winner forms the first singleton component of $\mathcal{P}_W$. 
Consequently, the \texttt{XCC} implies the \texttt{CC}. 

Figure~\ref{fig:cc-xcc-scc} illustrates the differences among the three criteria. 
In Figure~\ref{fig:cc-xcc-scc}(a), the \texttt{CC} requires the Condorcet winner $a$ to be ranked first, while the \texttt{XCC} further requires every candidate in $\{b,c,d\}$ to be ranked above $e$. 
The tied comparisons between $b$ and $d$ and between $c$ and $d$, represented in Figure~\ref{fig:cc-xcc-scc}(a) by
\begin{align*}
(b,d),(d,b),(c,d),(d,c)\in A_W,
\end{align*}
are absent from $G_W^+$ in Figure~\ref{fig:cc-xcc-scc}(b). 
Since $(b,c)\in A_W^+$ and $b$ and $c$ belong to distinct components of $\mathcal{P}_W^+$, the \texttt{SCC} further requires $b$ to be ranked above $c$.

\begin{figure*}[t]
\centering

\begin{subfigure}[t]{0.45\textwidth}
\centering
\begin{tikzpicture}[
    scale=0.9,
    every node/.style={transform shape},
    vertex/.style={circle,draw,minimum size=6mm,inner sep=0pt},
    strict/.style={-{Latex},black},
    tied/.style={{Latex}-{Latex},black},
    component/.style={dashed,rounded corners,black}
]
\node[vertex] (a) at (0,0) {$a$};
\node[vertex] (b) at (3,1.5) {$b$};
\node[vertex] (c) at (3,0) {$c$};
\node[vertex] (d) at (1.5,-1.5) {$d$};
\node[vertex] (e) at (4.5,-1.5) {$e$};
\draw[strict] (a) -- (b);
\draw[strict] (a) -- (c);
\draw[strict] (a) -- (d);
\draw[strict] (a) -- (e);
\draw[strict] (b) -- (e);
\draw[strict] (c) -- (e);
\draw[strict] (d) -- (e);
\draw[strict] (b) -- (c);
\draw[strict] (b) -- (c);
\draw[tied] (d) -- (b);
\draw[tied] (d) -- (c);
\draw[component] (-0.6,-0.6) rectangle (0.6,0.6);
\draw[component] (0.9,-2.1) rectangle (3.6,2.1);
\draw[component] (3.9,-2.1) rectangle (5.1,-0.9);
\end{tikzpicture}
\caption{The majority graph $G_W$.}
\label{fig:cc-xcc-gw}
\end{subfigure}
\begin{subfigure}[t]{0.45\textwidth}
\centering
\begin{tikzpicture}[
    scale=0.9,
    every node/.style={transform shape},
    vertex/.style={circle,draw,minimum size=6mm,inner sep=0pt},
    strict/.style={-{Latex},black},
    tied/.style={{Latex}-{Latex},black},
    component/.style={dashed,rounded corners,black}
]
\node[vertex] (a) at (0,0) {$a$};
\node[vertex] (b) at (3,1.5) {$b$};
\node[vertex] (c) at (3,0) {$c$};
\node[vertex] (d) at (1.5,-1.5) {$d$};
\node[vertex] (e) at (4.5,-1.5) {$e$};
\draw[strict] (a) -- (b);
\draw[strict] (a) -- (c);
\draw[strict] (a) -- (d);
\draw[strict] (a) -- (e);
\draw[strict] (b) -- (e);
\draw[strict] (c) -- (e);
\draw[strict] (d) -- (e);
\draw[strict] (b) -- (c);
\draw[component] (-0.6,-0.6) rectangle (0.6,0.6);
\draw[component] (2.4,0.9) rectangle (3.6,2.1);
\draw[component] (2.4,-0.6) rectangle (3.6,0.6);
\draw[component] (0.9,-2.1) rectangle (2.1,-0.9);
\draw[component] (3.9,-2.1) rectangle (5.1,-0.9);
\end{tikzpicture}
\caption{The strict majority graph $G_W^+$.}
\label{fig:cc-xcc-gwplus}
\end{subfigure}
\caption{Comparison of the \texttt{CC}, \texttt{XCC}, and \texttt{SCC}. 
The dashed regions represent strongly connected components. 
In (a), the \texttt{CC} requires the Condorcet winner $a$ to be ranked first. 
The \texttt{XCC} further requires every candidate in $\{b,c,d\}$ to be ranked higher than $e$.
The opposite arcs between $b$ and $d$ and between $c$ and $d$ represent tied pairwise comparisons and are therefore absent from $G_W^+$ in (b). 
Consequently, $b$ and $c$ belong to different components of $G_W^+$, and the arc $(b,c)\in A_W^+$ implies by the \texttt{SCC} that $b$ must be ranked higher than $c$.}
\label{fig:cc-xcc-scc}
\end{figure*}
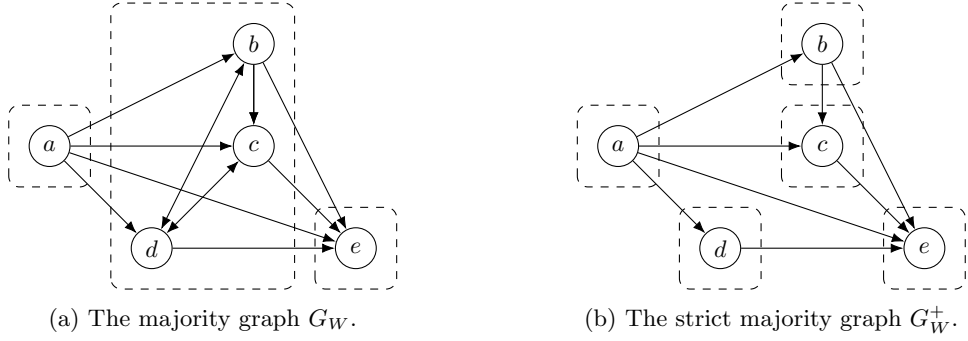

\subsection{Extension to rankings with ties}

A ranking with ties is a surjection $\pi:V\to[p]$ for some $p\le n$. 
We write 
\begin{align*}
u\sim_\pi v &\quad \text{if }\pi(u)=\pi(v),\\
u\succ_\pi v &\quad \text{if }\pi(u)<\pi(v),\\
u\succsim_\pi v &\quad \text{if }\pi(u)\le\pi(v). 
\end{align*}
The relation $\succsim_\pi$ defines a weak order. 
For any two rankings with ties $\pi,\rho$, their Kemeny--Snell distance is defined as
\begin{align}\label{KemenySnell}
d_{\mathrm{KS}}(\pi,\rho)
=
\sum_{\{u,v\}\in\mathcal U}
\left|\operatorname{sgn}(\pi(u)-\pi(v))
- \operatorname{sgn}(\rho(u)-\rho(v)) \right|.
\end{align}
Since rankings with ties are allowed, each sign function takes a value in $\{-1,0,1\}$.
Hence, the contribution of each candidate pair is zero if the two rankings specify the same relation, one if one ranking ties the pair and the other ranks it strictly, and two if the two rankings specify opposite strict orders.
Suppose that $\Pi$ is a finite collection of rankings with ties.
We call a ranking with ties $\pi$ that minimizes
\begin{align*}
f_\Pi^{\texttt{KPT}}(\pi)
=
\sum_{\rho \in \Pi}
d_{\mathrm{KS}}(\pi,\rho)
\end{align*}
over all rankings with ties on $V$ a Kemeny ranking with ties, and 
the problem of finding a Kemeny ranking with ties the Kemeny problem with ties (\texttt{KPT}).

When the rankings are strict, the Kemeny--Snell distance is twice the Kendall tau distance.
Hence, the \texttt{KPT} objective is twice the \texttt{KP} objective, and the two problems have the same optimal solutions.
In this sense, the \texttt{KPT} generalizes the \texttt{KP}.
Therefore, when ties are not allowed, the two objective functions have the same minimizers.
In this sense, the \texttt{KPT} generalizes the \texttt{KP}.

The linear ordering problem with ties (\texttt{LOPT}) generalizes the \texttt{KPT}.
For distinct candidates $u,v\in V$, let $w^\succ_{uv}$ denote the real-valued weight assigned when $u$ is ranked above $v$, and $w^\sim_{uv}$ denote the real-valued weight assigned when $u$ and $v$ are tied.
We use the convention 
$w^\sim_{uv}=w^\sim_{vu}$.
Let $W$ denote the collection of all these weights.
Given $W$, the \texttt{LOPT} is to find a ranking with ties $\pi$ that maximizes
\begin{align*}
f^{\texttt{LOPT}}_W(\pi)
\coloneqq
\sum_{\{u,v\}\in\mathcal U}
\left(
w^\succ_{uv}\mathbf{1}[u\succ_\pi v]
+
w^\succ_{vu}\mathbf{1}[v\succ_\pi u]
+
w^\sim_{uv}\mathbf{1}[u\sim_\pi v]
\right).
\end{align*}
We denote this instance by $\texttt{LOPT}(W)$.

We next show that the \texttt{KPT} is a special case of the \texttt{LOPT}.
For each $\{u,v\}\in\mathcal U$, let $p_{uv}$ and $p_{vu}$ denote the numbers of input rankings that rank $u$ above $v$ and $v$ above $u$, respectively.
Let $t_{uv}$ denote the number of input rankings that tie $u$ and $v$.
Define
\begin{align*}
w^\succ_{uv}
&\coloneqq
p_{uv}-p_{vu},\\
w^\succ_{vu}
&\coloneqq
p_{vu}-p_{uv},\\
w^\sim_{uv}
&\coloneqq
t_{uv}.
\end{align*}
To see why these weights represent the \texttt{KPT}, consider the contribution of a pair $\{u,v\}$ to the total Kemeny--Snell distance.
For each pair $\{u,v\}\in\mathcal U$, its contribution to the total Kemeny--Snell distance is
\begin{align*}
2p_{vu}+t_{uv}&\quad\text{if }u\succ_\pi v,\\
p_{uv}+p_{vu}&\quad\text{if }u\sim_\pi v,\\
2p_{uv}+t_{uv}&\quad\text{if }v\succ_\pi u.
\end{align*}
Adding the corresponding contribution to the \texttt{LOPT} objective gives
$p_{uv}+p_{vu}+t_{uv}=|\Pi|$ in each case.
Thus, the sum is $|\Pi|$ for every pair, regardless of the relation selected by $\pi$.
This implies that for every candidate pair, the sum of its contribution to the Kemeny--Snell distance and its contribution to the \texttt{LOPT} objective is independent of the relation selected by $\pi$.
Consequently, for every ranking with ties $\pi$, we have
\begin{align*}
f_\Pi^{\texttt{KPT}}(\pi)+f^\texttt{LOPT}_W(\pi)=|\Pi||\mathcal U|.
\end{align*}
Therefore, a ranking with ties minimizes the total Kemeny--Snell distance if and only if it maximizes the corresponding \texttt{LOPT} objective.
Hence, the optimal solution set of the \texttt{KPT} coincides with that of the corresponding \texttt{LOPT} instance.

We may normalize the three weights associated with each candidate pair without changing the optimal solutions.

\begin{observation}\label{obs:normalization}
For every instance $\texttt{LOPT}(W)$, we may assume without loss of generality that
\begin{align*}
\min\{w^\succ_{uv},w^\succ_{vu},w^\sim_{uv}\}=0
\end{align*}
for every $\{u,v\}\in\mathcal U$.
In particular, all weights are nonnegative.
\end{observation}

Indeed, for each $\{u,v\}\in\mathcal U$, subtracting the minimum of $w^\succ_{uv}$, $w^\succ_{vu}$, and $w^\sim_{uv}$ from all three weights changes the objective value of every ranking with ties by the same constant.
Therefore, this operation does not change the optimal solutions.
We use this normalization throughout the remainder of the paper.

\subsection{An extension of the majority graph}

To state our main results, we first extend the majority graph to associate a directed graph with each \texttt{LOPT} instance. 


\begin{definition}[Majority graph for the \texttt{LOPT}]
For any instance $\texttt{LOPT}(W)$, define the extended majority graph
$H_W=(V,B_W)$ as follows.
For each pair $\{x,y\}\in\mathcal U$, the arcs are included according to the following conditions:
\begin{enumerate}
\item[(C1)]
$(x,y)\in B_W$ if $w^\succ_{xy}>\max\{w^\succ_{yx},w^\sim_{xy}\}$.
\item[(C2)]
$(y,x)\in B_W$ if $w^\succ_{yx}>\max\{w^\succ_{xy},w^\sim_{xy}\}$.
\item[(C3)]
Both $(x,y)$ and $(y,x)$ belong to $B_W$ if $w^\sim_{xy}>\max\{w^\succ_{xy},w^\succ_{yx}\}$.
\item[(C4)]
Both $(x,y)$ and $(y,x)$ belong to $B_W$ if exactly two of $w^\succ_{xy}$, $w^\succ_{yx}$, and $w^\sim_{xy}$ attain the maximum.
\end{enumerate}
Let $\mathcal Q_W$ denote the strongly connected component decomposition of $H_W$.
\end{definition}

The graph $H_W$ contains neither $(x,y)$ nor $(y,x)$ precisely when the three weights associated with $\{x,y\}$ are equal.
By Observation~\ref{obs:normalization}, all three weights are zero in this case.
Therefore, the relation selected for this pair does not affect the objective value.
Note that $H_W$ need not be semicomplete.
Omitting such irrelevant candidate pairs may produce a finer strongly connected component decomposition.

\section{Main results}\label{MainResults}

We now state our main results.
The first result extends the \texttt{SCC} to the \texttt{LOPT}.

\begin{theorem}[Non-strict strong Condorcet criterion]\label{thm:strict}
For any instance $\texttt{LOPT}(W)$, let $u,v\in V$ belong to distinct components in $\mathcal Q_W$.
If $(u,v)\in B_W$, then $\pi(u)<\pi(v)$ for every optimal solution $\pi$ to $\texttt{LOPT}(W)$.
\end{theorem}

When ties are not allowed, the \texttt{LOPT} reduces to the \texttt{LOP}.
In this case, $H_W$ reduces to the strict majority graph for the \texttt{LOP}, and Theorem~\ref{thm:strict} reduces to the \texttt{SCC}.
Our second result concerns ties within strongly connected components.
We call a component $Y\in\mathcal Q_W$ 
tie-dominant 
if
\begin{align*}
w^\sim_{xy}
\geq
\max\{w^\succ_{xy},w^\succ_{yx}\}
\end{align*}
for every $\{x,y\}\in\mathcal U$ with $x,y\in Y$.

\begin{theorem}[Tie property]\label{thm:tie}
For every instance $\texttt{LOPT}(W)$, there exists an optimal solution that ties all candidates within every tie-dominant component in $\mathcal Q_W$.
Moreover, let $u,v\in V$ be distinct candidates that belong to the same tie-dominant component.
If
\begin{align*}
w^\sim_{uv}
>
\max\{w^\succ_{uv},w^\succ_{vu}\},
\end{align*}
then $\pi(u)=\pi(v)$ for every optimal solution $\pi$ to $\texttt{LOPT}(W)$.
\end{theorem}

Theorem~\ref{thm:strict} concerns direct arcs between distinct components.
The following corollary also determines comparisons between candidate pairs that need not be joined by an arc of $H_W$.

\begin{corollary}\label{cor:component-structure}
For every instance $\texttt{LOPT}(W)$, there exists an optimal solution $\pi$ to $\texttt{LOPT}(W)$ satisfying the following conditions:
\begin{enumerate}
\item
For every $s,t\in V$ belonging to the same tie-dominant component of $\mathcal{Q}_W$, we have $\pi(s)=\pi(t)$.
\item
For every $s,t\in V$ belonging to distinct components of $\mathcal{Q}_W$, if $H_W$ contains 
a directed path 
from $s$ to $t$, then we have $\pi(s)<\pi(t)$.
\end{enumerate}
\end{corollary}

\begin{proof}
By Theorem~\ref{thm:tie}, there exists an optimal solution $\pi^*$ that ties all candidates within every tie-dominant component.
Fix a topological order of $\mathcal Q_W$. 
Let $\pi$ be a ranking with ties that keeps the relative order given by $\pi^*$ within each component and places every candidate in an earlier component above every candidate in a later component in $\mathcal Q_W$. 
The first condition holds by construction.
The second condition also holds since, for any $s,t\in V$ satisfying its assumptions, the component containing $s$ precedes the component containing $t$ in every topological order, and hence $\pi(s)<\pi(t)$. 

It remains to show that $\pi$ is optimal.
The rankings $\pi$ and $\pi^*$ determine the same relation for every pair within the same component.
Therefore, the contribution of every such pair remains unchanged.
Next, let $x$ and $y$ belong to distinct components, where the component containing $x$ precedes the component containing $y$ in the chosen topological order. 
Suppose that $(x,y)\in B_W$ or $(y,x)\in B_W$.
The topological order implies that $(x,y)\in B_W$.
By Theorem~\ref{thm:strict}, $\pi^*(x)<\pi^*(y)$. 
On the other hand, by construction, $\pi(x)<\pi(y)$.
The contribution of $\{x,y\}$ therefore remains unchanged. 
Finally, suppose that neither $(x,y),(y,x) \notin B_W$.
Then, all three weights associated with $\{x,y\}$ are zero.
Therefore, the relation between $x$ and $y$ does not affect the objective value. 
It follows that $\pi$ is optimal. 
\end{proof}

\begin{table}[t]
\centering
\caption{Weights and arcs in Example~\ref{ex:two-tie-cases}. 
Bold values indicate the maximum weight for each candidate pair.}
\label{tab:two-tie-cases}
\begin{tabular}{cccccl}
\toprule
$x$
& $y$
& $w^\succ_{xy}$
& $w^\succ_{yx}$
& $w^\sim_{xy}$
& Arcs in $B_W$\\
\midrule
$a$ & $b$ & $1$          & $1$ & $\mathbf{3}$ & $(a,b),(b,a)$\\
$a$ & $c$ & $\mathbf{3}$ & $1$ & $1$          & $(a,c)$\\
$a$ & $d$ & $\mathbf{3}$ & $1$ & $1$          & $(a,d)$\\
$b$ & $c$ & $\mathbf{3}$ & $1$ & $1$          & $(b,c)$\\
$b$ & $d$ & $\mathbf{3}$ & $1$ & $1$          & $(b,d)$\\
$c$ & $d$ & $\mathbf{2}$ & $1$ & $\mathbf{2}$ & $(c,d),(d,c)$\\
\bottomrule
\end{tabular}
\end{table}

\begin{figure}[t]
\centering
\begin{tikzpicture}[
    scale=0.9,
    every node/.style={transform shape},
    vertex/.style={
        circle,
        draw,
        minimum size=6mm,
        inner sep=0pt
    },
    strict/.style={
        -{Latex},
        black
    },
    tied/.style={
        {Latex}-{Latex},
        black
    },
    component/.style={
        dashed,
        rounded corners,
        black
    }
]
\node[vertex] (a) at (0,1) {$a$};
\node[vertex] (b) at (0,-1) {$b$};
\node[vertex] (c) at (3,1) {$c$};
\node[vertex] (d) at (3,-1) {$d$};
\draw[tied] (a) -- (b);
\draw[tied] (c) -- (d);
\draw[strict] (a) -- (c);
\draw[strict] (a) -- (d);
\draw[strict] (b) -- (c);
\draw[strict] (b) -- (d);
\draw[component]    (-0.6,-1.6) rectangle (0.6,1.6);
\draw[component]    (2.4,-1.6) rectangle (3.6,1.6);
\end{tikzpicture}
\caption{The extended majority graph $H_W$ for Example~\ref{ex:two-tie-cases}.}
\label{fig:two-tie-cases}
\end{figure}
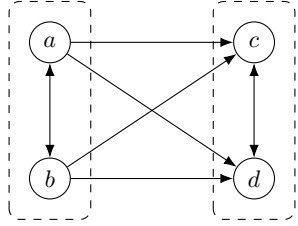

\begin{example}\label{ex:two-tie-cases}
Let $V=\{a,b,c,d\}$, and assign the weights shown in Table~\ref{tab:two-tie-cases}. 
The resulting extended majority graph is shown in Figure~\ref{fig:two-tie-cases}. 
The set of the strongly connected components is 
\begin{align*}
\mathcal Q_W = \{V_1,V_2\} = \{\{a,b\}, \{c,d\}\}.
\end{align*}

We first apply Theorem~\ref{thm:strict}. 
For every $x\in\{a,b\}$ and $y\in\{c,d\}$, we have $(x,y)\in B_W$ and $(y,x)\notin B_W$.
Theorem~\ref{thm:strict} therefore implies that, for every optimal solution $\pi$, we have $\pi(x)<\pi(y)$ for every $x\in\{a,b\}$ and $y\in\{c,d\}$. 
This is the ordering between components determined by the \texttt{NSCC}.
The existing \texttt{SCC} gives the corresponding conclusion for strict rankings but does not describe relations within the components when ties are allowed.

We next apply Theorem~\ref{thm:tie}.
Both $\{a,b\}$ and $\{c,d\}$ are tie-dominant.
The first assertion of Theorem~\ref{thm:tie} therefore guarantees an optimal solution that ties the candidates within both components simultaneously.
For the pair $\{a,b\}$, the tie weight is the unique maximum.
The second assertion of Theorem~\ref{thm:tie} therefore implies that every optimal solution $\pi$ satisfies $\pi(a)=\pi(b)$.
For the pair $\{c,d\}$, the tie weight shares the maximum value with the weight for $c\succ d$.
Therefore, Theorem~\ref{thm:tie} does not require $c$ and $d$ to be tied in every optimal solution.
Indeed, both
\begin{align*}
a\sim b\succ c\sim d
\quad\text{and}\quad
a\sim b\succ c\succ d
\end{align*}
are optimal, and each attains the objective value $17$.
\end{example}

\section{Graph formulation}\label{GraphReformulation}

Ando et al.~\cite{Andoetal2022} reduced the \texttt{LOP} to the weighted feedback arc set problem to prove the \texttt{SCC}.
We also use a graph-theoretic approach for the \texttt{LOPT}.

\subsection{Graph representation of rankings with ties}

It is well known that every acyclic directed graph admits a linear extension.
This fact connects acyclic directed graphs with strict rankings.
When ties are allowed, opposite arcs represent ties, and hence the graph need not be acyclic.
We therefore use the following weaker condition. 

\begin{definition}[Unicycle-free directed graph]
Let $G=(V,A)$ be a directed graph.
A directed cycle 
\begin{align*}
C=(v_1,v_2,\ldots,v_k,v_1)
\end{align*}
has distinct vertices $v_1,v_2,\ldots,v_k$.
We say that the cycle $C$ is a unicycle if its reverse $(v_1,v_k,\ldots,v_2,v_1)$ is not contained in $G$.
The graph $G$ is unicycle-free if it contains no unicycle.
\end{definition}

\begin{definition}[Parsimonious extension]\label{def:parsimonious-extension}
Let $G=(V,A)$ be a directed graph.
A ranking with ties $\pi$ is a linear extension of $G$ if $\pi(x)\leq\pi(y)$ for every $(x,y)\in A$.
A linear extension $\pi$ is parsimonious if $\pi(x)<\pi(y)$ whenever $(x,y)\in A$ and $(y,x)\notin A$.
\end{definition}

\begin{observation}\label{lem:parsimonious-extension}
A directed graph admits a parsimonious extension if and only if it is unicycle-free.
\end{observation}

\begin{figure}[t]
\centering
\begin{subfigure}[t]{0.3\textwidth}
\centering
\begin{tikzpicture}[
    scale=0.9,
    every node/.style={transform shape},
    vertex/.style={circle,draw,minimum size=6mm,inner sep=0pt},
    strict/.style={-{Latex},black},
    tied/.style={{Latex}-{Latex},black},
    component/.style={dashed,rounded corners,black}
]
\node[vertex] (a) at (0,2) {$a$};
\node[vertex] (b) at (0,0) {$b$};
\node[vertex] (c) at (2,0) {$c$};
\draw[strict] (a) -- (b);
\draw[strict] (a) -- (c);
\draw[tied] (b) -- (c);
\end{tikzpicture}
\caption{Unicycle-free}
\label{fig:unicycle-a}
\end{subfigure}
\begin{subfigure}[t]{0.3\textwidth}
\centering
\begin{tikzpicture}[
    scale=0.9,
    every node/.style={transform shape},
    vertex/.style={circle,draw,minimum size=6mm,inner sep=0pt},
    strict/.style={-{Latex},black},
    tied/.style={{Latex}-{Latex},black},
    component/.style={dashed,rounded corners,black}
]
\node[vertex] (a) at (0,2) {$a$};
\node[vertex] (b) at (0,0) {$b$};
\node[vertex] (c) at (2,0) {$c$};
\draw[tied] (a) -- (b);
\draw[tied] (a) -- (c);
\draw[tied] (b) -- (c);
\end{tikzpicture}
\caption{Unicycle-free}
\label{fig:unicycle-b}
\end{subfigure}
\begin{subfigure}[t]{0.3\textwidth}
\centering
\begin{tikzpicture}[
    scale=0.9,
    every node/.style={transform shape},
    vertex/.style={circle,draw,minimum size=6mm,inner sep=0pt},
    strict/.style={-{Latex},black},
    tied/.style={{Latex}-{Latex},black},
    component/.style={dashed,rounded corners,black}
]
\node[vertex] (a) at (0,2) {$a$};
\node[vertex] (b) at (0,0) {$b$};
\node[vertex] (c) at (2,0) {$c$};
\draw[tied] (a) -- (b);
\draw[strict] (a) -- (c);
\draw[tied] (b) -- (c);
\end{tikzpicture}
\caption{Not unicycle-free}
\label{fig:unicycle-c}
\end{subfigure}
\caption{Examples of parsimonious extensions.
The graphs in (a) and (b) admit the parsimonious extensions $a\succ b\sim c$ and $a\sim b\sim c$, respectively.
The graph in (c) contains the unicycle $(a,c,b,a)$ and therefore admits no parsimonious extension.}
\label{fig:unicycle-examples}
\end{figure}
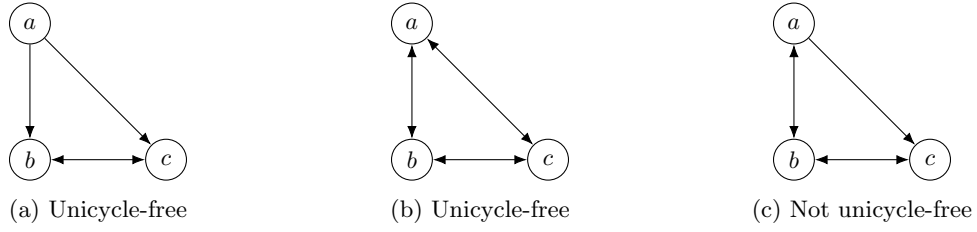

For completeness, the proof is given in \ref{app:graph-representation}.
Note that our use of the term unicycle differs from the standard notion of a unicyclic graph.
Observation~\ref{lem:parsimonious-extension} allows us to obtain rankings with ties from directed graphs that are not semicomplete.

\subsection{An equivalent graph optimization problem}

We now formulate the \texttt{LOPT} as a graph optimization problem.
Let $G=(V,A)$ be a directed graph.
Let $F$ be an arc set which does not need to be a subset of $A$.
Define 
\begin{align*}
\operatorname{add}(F)=F\setminus A
\end{align*}
The set $\operatorname{add}(F)$ contains the arcs added to obtain $F$ from $A$.

\begin{definition}[Feasible arc set]\label{def:feasible-arc-set}
An arc set $F$ is feasible for $G$ if the following conditions hold:
\begin{enumerate}
\item[(F1)]
$(V,F)$ is unicycle-free.
\item[(F2)]
For every $(x,y)\in\operatorname{add}(F)$, we have $(y,x)\in A$.
\item[(F3)]
For every pair of distinct vertices $x,y\in V$, if $(x,y)\in A$ or $(y,x)\in A$, then $(x,y)\in F$ or $(y,x)\in F$.
\end{enumerate}
\end{definition}

Let $F$ be a feasible arc set.
For every $\{x,y\}\in\mathcal U$, we write
\begin{align*}
x\succ_F y
&\quad\text{if }(x,y)\in F\text{ and }(y,x)\notin F,\\
x\sim_F y
&\quad\text{if }(x,y),(y,x)\in F, 
\end{align*}
and given weights $W$, define
\begin{align*}
g_W(F)
=
\sum_{\{x,y\}\in\mathcal U}
\left(
w^\succ_{xy}\mathbf{1}[x\succ_F y]
+
w^\succ_{yx}\mathbf{1}[y\succ_F x]
+
w^\sim_{xy}\mathbf{1}[x\sim_F y]
\right).
\end{align*}


\begin{definition}[Weighted unicycle-free graph problem]
Given a directed graph $G=(V,A)$ and weights $W$, the weighted unicycle-free graph problem (\texttt{WUCF}) seeks a feasible arc set $F$ that maximizes $g_W(F)$.
We denote this instance by $\texttt{WUCF}(G,W)$.
\end{definition}

For an instance $\texttt{LOPT}(W)$, we focus on the extended majority graph $H_W$ and consider the corresponding instance $\texttt{WUCF}(H_W,W)$.
We associate solutions of the two problems in both directions.

We first associate an arc set $F_\pi$ with each ranking with ties $\pi$.
It follows from the definition of $B_W$ and Observation~\ref{obs:normalization} that we have $(x,y)\in B_W$ or $(y,x) \in B_W$ if and only if at least one of $w^\succ_{xy}$, $w^\succ_{yx}$, and $w^\sim_{xy}$ is positive.
Thus, only candidate pairs satisfying this condition can make a nonzero contribution to the \texttt{LOPT} objective.
We therefore define 
\begin{align*}
F_\pi
=
\{(x,y):
\{x,y\}\in\mathcal U,\ 
\pi(x)\leq\pi(y),\ 
(x,y)\in B_W\text{ or }(y,x)\in B_W
\}.
\end{align*}

Conversely, let $F$ be a feasible arc set for $H_W$.
By \textnormal{(F1)} and Observation~\ref{lem:parsimonious-extension}, $(V,F)$ admits a parsimonious extension.
We associate with $F$ any parsimonious extension of $(V,F)$ and denote it by $\pi_F$.
The following lemma establishes the relation between the optimal solutions of the two problems.

\begin{lemma}\label{lem:lopt-wucf}
For every ranking with ties $\pi$, the arc set $F_\pi$ is feasible for $H_W$.
Conversely, for every feasible arc set $F$ for $H_W$, the graph $(V,F)$ admits a parsimonious extension.
Let $\pi_F$ be any such extension.
Then, we have
\begin{align*}
f_W^{\texttt{LOPT}}(\pi)
=
g_W(F_\pi) \quad \text{and} \quad
g_W(F)
=
f_W^{\texttt{LOPT}}(\pi_F).
\end{align*}
Consequently, $\pi$ is optimal for $\texttt{LOPT}(W)$ if and only if $F_\pi$ is optimal for $\texttt{WUCF}(H_W,W)$.
Moreover, if $F$ is optimal for $\texttt{WUCF}(H_W,W)$, then every parsimonious extension $\pi_F$ of $(V,F)$ is optimal for $\texttt{LOPT}(W)$.
\end{lemma}

\begin{proof}
Let $\pi$ be a ranking with ties.
The ranking $\pi$ is a parsimonious extension of $(V,F_\pi)$.
Observation~\ref{lem:parsimonious-extension} implies that $(V,F_\pi)$ is unicycle-free.
Suppose that $(x,y)\in\operatorname{add}(F_\pi)$.
Then, $(x,y)\notin B_W$.
By the definition of $F_\pi$, we have $(y,x)\in B_W$.
Suppose that $(x,y)\in B_W$ or $(y,x)\in B_W$.
Since $\pi$ is complete, we have
\begin{align*}
\pi(x)\leq\pi(y) \quad \text{or} \quad \pi(y)\leq\pi(x).
\end{align*}
Therefore, $(x,y)\in F_\pi$ or $(y,x)\in F_\pi$.
Hence, $F_\pi$ is feasible for $H_W$.

Consider a pair $\{x,y\}\in\mathcal U$ such that $(x,y)\in B_W$ or $(y,x)\in B_W$.
By the definition of $F_\pi$, the ranking $\pi$ and the arc set $F_\pi$ determine the same relation for this pair.
Hence, this pair has the same contribution to the two objectives.
Next, consider a pair $\{x,y\}\in\mathcal U$ such that $(x,y),(y,x) \notin B_W$.
By the definition of $B_W$ and Observation~\ref{obs:normalization}, all three weights associated with this pair are zero.
Hence, this pair also has the same contribution, namely zero, to the two objectives. Therefore, 
\begin{align*}
f_W^{\texttt{LOPT}}(\pi)
=
g_W(F_\pi).
\end{align*}

Conversely, let $F$ be feasible for $H_W$.
By \textnormal{(F1)}, $(V,F)$ is unicycle-free.
Observation~\ref{lem:parsimonious-extension} implies that $(V,F)$ admits a parsimonious extension.
Let $\pi_F$ be any such extension. 
Then, $\pi_F$ is feasible for $\texttt{LOPT}(W)$.

Consider a pair $\{x,y\}\in\mathcal U$ such that $(x,y)\in B_W$ or $(y,x)\in B_W$.
Condition \textnormal{(F3)} implies that $(x,y)\in F$ or $(y,x) \in F$.
By the linear-extension and parsimonious-extension conditions, we have
\begin{align*}
(x,y)\text{ and }(y,x)\in F
&\quad\Rightarrow\quad
\pi_F(x)=\pi_F(y),\\
(x,y)\in F\text{ and }(y,x)\notin F
&\quad\Rightarrow\quad
\pi_F(x)<\pi_F(y),\\
(x,y)\notin F\text{ and }(y,x)\in F
&\quad\Rightarrow\quad
\pi_F(y)<\pi_F(x).
\end{align*}
Therefore, $F$ and $\pi_F$ determine the same relation for every such pair.
Next, consider a pair $\{x,y\}\in\mathcal U$ such that $(x,y),(y,x)\notin B_W$.
As before, by the definition of $B_W$ and Observation~\ref{obs:normalization}, all three weights associated with this pair are zero.
Hence, this pair contributes zero to both objective functions.
It follows that
\begin{align*}
g_W(F)=f_W^{\texttt{LOPT}}(\pi_F).
\end{align*}

Every feasible solution of either problem is associated with a feasible solution of the other problem having the same objective value.
Therefore, the two problems have the same optimal objective value.
The desired result follows.
\end{proof}

\section{Proofs of the main results}\label{ProofsMainResults}

We prove the main theorems by modifying optimal feasible arc sets between and within strongly connected components of $H_W$.

\subsection{Proof of Theorem~\ref{thm:strict}}


To prove Theorem~\ref{thm:strict}, it suffices to establish the following claim.

\begin{claim}\label{lem:intercomponent-arcs}
Let $F$ be an optimal solution to $\texttt{WUCF}(H_W,W)$.
If $(x,y)\in B_W$ and $x$ and $y$ belong to distinct strongly connected components of $\mathcal{Q}_W$, then
\begin{align*}
(x,y)\in F
\quad\text{and}\quad
(y,x)\notin F
\end{align*}
\end{claim}

Indeed, Claim~\ref{lem:intercomponent-arcs} implies Theorem~\ref{thm:strict} as follows. 
Let $u$ and $v$ belong to distinct components of $H_W$, and suppose that $(u,v)\in B_W$.
Let $\pi$ be any optimal solution to $\texttt{LOPT}(W)$.
Then, by Lemma~\ref{lem:lopt-wucf}, $F_\pi$ is an optimal solution to $\texttt{WUCF}(H_W,W)$.
Claim~\ref{lem:intercomponent-arcs} gives $(u,v)\in F_\pi$ and $(v,u)\notin F_\pi$.
Therefore, $u\succ_\pi v$.

\begin{proof}
Let $F$ be an optimal solution to $\texttt{WUCF}(H_W,W)$.
Let $\mathcal{C}$ and $\mathcal{D}$ denote the sets of arcs in $B_W$ and $F$, respectively, whose endpoints belong to distinct components in $\mathcal Q_W$.
Define
$
F'
=
(F\setminus\mathcal D)\cup\mathcal C$.
Thus, $F'$ and $F$ contain the same arcs between candidates in the same component.
For candidates in distinct components, $F'$ contains exactly the arcs of $B_W$. 
We show that $F'$ is feasible for $H_W$. 

We first show that every directed cycle of $(V,F')$ is contained in a single component in $\mathcal Q_W$.
Suppose to the contrary that there exists a directed cycle $C$ that visits distinct components. 
Then, every arc of $C$ between distinct components belongs to $B_W$ by the definition of $F'$. 
Therefore, the components visited by $C$ must be mutually reachable in $H_W$, a contradiction. 
On the other hand, $F'$ and $F$ have the same arcs within each component in $\mathcal Q_W$.
Therefore, $(V,F')$ is unicycle-free and satisfies \textnormal{(F1)}.

It remains to show that \textnormal{(F2)} and \textnormal{(F3)} hold for $F'$. 
By the construction of $F'$, every arc in $F'\setminus B_W$ also belongs to $F\setminus B_W$; that is,
\begin{align*}
F'\setminus B_W
\subseteq
F\setminus B_W.
\end{align*}
Therefore, since $F$ satisfies \textnormal{(F2)}, $F'$ also satisfies  \textnormal{(F2)}.
To see that \textnormal{(F3)} also holds, let $x$ and $y$ be distinct vertices such that $(x,y)\in B_W$ or $(y,x)\in B_W$. 
If $x$ and $y$ belong to the same component in $\mathcal Q_W$, then $F'$ and $F$ contain the same arcs between them. 
Thus, \textnormal{(F3)} for $F$ implies that $(x,y)\in F'$ or $(y,x)\in F'$. 
If $x$ and $y$ belong to distinct components, then every arc of $B_W$ between them is included in $F'$. 
Therefore, we have $(x,y)\in F'$ or $(y,x)\in F'$. 
Thus, $F'$ satisfies \textnormal{(F3)}. 

Next, we compare the objective values of $F$ and $F'$.
Consider any pair $\{x,y\}\in\mathcal U$.
Suppose first that $x$ and $y$ belong to the same component in $\mathcal Q_W$.
By the definition of $F'$, the arc sets $F$ and $F'$ determine the same relation for this pair.
Therefore, its contribution does not change.
Next, suppose that $x$ and $y$ belong to distinct components in $\mathcal Q_W$. 

\paragraph{Case 1}
Suppose that 
$(x,y),(y,x)\notin B_W$.
By the definition of $B_W$ and Observation~\ref{obs:normalization}, all three weights associated with this pair are zero.
Therefore, the pair contributes zero under both $F$ and $F'$.

\paragraph{Case 2}
Suppose that 
$(x,y)\in B_W$ or $(y,x)\in B_W$. 
Since $x$ and $y$ belong to distinct components, the two arcs cannot both belong to $B_W$. 
Thus, after exchanging $x$ and $y$ if necessary, we can assume that $(x,y)\in B_W$ and $(y,x)\notin B_W$.
Then, by the definition of $F'$, we have $(x,y)\in F'$ and $(y,x)\notin F'$.
Hence, $x\succ_{F'}y$.
On the other hand, the definition of $B_W$ implies that
\begin{align*}
w^\succ_{xy}
>
\max\{w^\succ_{yx},w^\sim_{xy}\}.
\end{align*}
Therefore, the contribution of $\{x,y\}$ under $F'$ is at least its contribution under $F$.
The inequality is strict unless $x\succ_F y$.
Thus, no pair has a smaller contribution under $F'$ than under $F$.
It follows that
\begin{align*}
g_W(F')\geq g_W(F).
\end{align*}

Suppose that the conclusion of the claim fails.
Then, there exist vertices $x$ and $y$ belonging to distinct components in $\mathcal Q_W$ such that $(x,y)\in B_W$ but $x\succ_F y$ does not hold.
By Case~2, the contribution of this pair $\{x,y\}$ strictly increases from $F$ to $F'$, whereas the contribution of every other pair does not decrease.
This contradicts the optimality of $F$ for $\texttt{WUCF}(H_W,W)$.
Therefore, the claim follows.
\end{proof}

\subsection{Proof of Theorem~\ref{thm:tie}}

To prove Theorem~\ref{thm:tie}, it suffices to establish the following claim.

\begin{claim}\label{claim:symmetric-component}
There exists an optimal solution $F$ to $\texttt{WUCF}(H_W,W)$ such that any parsimonious extension $\pi_F$ of $(V,F)$ ties all candidates within every tie-dominant component. 
Moreover, for any distinct candidates $u,v\in V$ belonging to the same tie-dominant component, if
\begin{align*}
w^\sim_{uv}
>
\max\{w^\succ_{uv},w^\succ_{vu}\},
\end{align*}
then every optimal solution  to $\texttt{WUCF}(H_W,W)$ contains both $(u,v)$ and $(v,u)$.
\end{claim}

Claim~\ref{claim:symmetric-component} implies Theorem~\ref{thm:tie} as follows. 
Claim~\ref{claim:symmetric-component} implies that there exists a parsimonious extension $\pi_F$ that ties all candidates within every tie-dominant component. 
Since $\pi_F$ is optimal for $\texttt{LOPT}(W)$ by Lemma~\ref{lem:lopt-wucf}, this implies the first assertion of Theorem~\ref{thm:tie}. 
Next, let $u$ and $v$ be distinct candidates in the same tie-dominant component, and suppose that
$w^\sim_{uv}
>
\max\{w^\succ_{uv},w^\succ_{vu}\}$ holds.  
Let $\pi$ be any optimal solution to $\texttt{LOPT}(W)$.
By Lemma~\ref{lem:lopt-wucf}, $F_\pi$ is optimal for $\texttt{WUCF}(H_W,W)$.
Claim~\ref{claim:symmetric-component} implies that 
$(u,v),(v,u)\in F_\pi$. 
Then, by the definition of $F_\pi$, we have $\pi(u)=\pi(v)$.
Since $\pi$ was arbitrary, every optimal solution to $\texttt{LOPT}(W)$ ties $u$ and $v$.

\begin{proof}
Let $Y\in\mathcal Q_W$ be any tie-dominant component.
Define
\begin{align*}
\mathcal T
=
\{(x,y):(x,y)\in B_W, \, x,y\in Y\}.
\end{align*}
Since $Y$ is tie-dominant, by the definition of $B_W$, we have
$(x,y)\in\mathcal T$ if and only if $(y,x)\in\mathcal T$. 
Let $F$ be any optimal solution to $\texttt{WUCF}(H_W,W)$.
Define $F'=F\cup\mathcal T$. 
We show that $F'$ is feasible for $H_W$. 

We readily see that \textnormal{(F2)} and \textnormal{(F3)} hold for $F'$. 
Since $\mathcal T\subseteq B_W$, we have
$F'\setminus B_W
=
F\setminus B_W$. 
Since \textnormal{(F2)} holds for $F$, this implies that \textnormal{(F2)} also holds for $F'$.   
Moreover, we have $F\subseteq F'$. 
Therefore, \textnormal{(F3)} holds for $F'$ since \textnormal{(F3)} holds for $F$. 

It remains to show that \textnormal{(F1)} holds for $F'$.
By Claim~\ref{lem:intercomponent-arcs} and \textnormal{(F2)}, $F$ and $B_W$ have the same arcs between distinct components in $\mathcal Q_W$.
Since $\mathcal T$ contains only arcs whose endpoints belong to $Y$, for all $x,y\in V$ that belong to distinct components in $\mathcal Q_W$, we have
\begin{align*}
(x,y)\in F'
\quad\text{if and only if}\quad
(x,y)\in B_W.
\end{align*}
Suppose that a directed cycle of $(V,F')$ visits at least two distinct components.
Since every arc of the cycle between distinct components belongs to $B_W$,  the components visited by the cycle are mutually reachable in $H_W$, contradicting the definition of strongly connected components. 
Thus, every directed cycle of $(V,F')$ is contained in a single component in $\mathcal Q_W$.
We, however, show that the subgraph induced by each component is unicycle-free, showing that \textnormal{(F1)} holds for $F'$ and hence $F'$ is feasible for $H_W$. 

\paragraph{Case 1}
Consider a component other than $Y$.
The arc sets $F'$ and $F$ are the same within this component.
Since $(V,F)$ is unicycle-free, the subgraph of $(V,F')$ induced by this component is also unicycle-free.

\paragraph{Case 2}
Consider the component $Y$.
Let $(x,y)\in F'$ with $x,y\in Y$.
We show that $(y,x)\in F'$.
\begin{itemize}
\item 
Suppose that $(x,y)\in\mathcal T$.  
Then, the symmetry of $\mathcal T$ implies that $(y,x)\in\mathcal T\subseteq F'$.
\item 
Suppose that $(x,y)\notin\mathcal T$. 
Then, $(x,y)\in F$. 
Moreover, $(x,y)\notin B_W$ by the definition of $\mathcal T$.
Thus, we have $(x,y)\in F\setminus B_W$.
Then, $(y,x)\in B_W$ because $F$ satisfies \textnormal{(F2)}.
Since $x,y\in Y$, by the definition of $\mathcal T$, this implies that $(y,x)\in\mathcal T\subseteq F'$.
\end{itemize}
Therefore, every arc of $F'$ within $Y$ is accompanied by its reverse.
Hence, the subgraph induced by $Y$ is unicycle-free.

Next, we show that the objective value of $F'$ is at least that of $F$. 
For this purpose, it suffices to look at the contributions of the pairs $x,y\in Y$.
Suppose that $(x,y),(y,x) \notin B_W$. 
Then, all three weights associated with the pair are zero by the definition of $B_W$ and Observation~\ref{obs:normalization}.
Therefore, the pair contributes zero under both $F$ and $F'$. 
Next, suppose that 
\begin{align*}
(x,y)\in B_W \quad \text{or} \quad (y,x)\in B_W.
\end{align*}
The definition of a tie-dominant component and the definition of $B_W$ imply that both arcs belong to $B_W$.
Hence, $(x,y),(y,x)\in F'$, and we have $x\sim_{F'}y$.
Since $Y$ is tie-dominant, we have
\begin{align*}
w^\sim_{xy}
\geq
\max\{w^\succ_{xy},w^\succ_{yx}\}.
\end{align*}
Therefore, the contribution of $\{x,y\}$ under $F'$ is at least its contribution under $F$.
Thus, $g_W(F')\geq g_W(F)$. 
Since $F$ is optimal and $F'$ is feasible, $F'$ is also optimal for $\texttt{WUCF}(H_W,W)$. 

Since $F'$ is feasible for $H_W$, Lemma~\ref{lem:lopt-wucf} ensures that $(V,F')$ admits a parsimonious extension.
Let $\pi_{F'}$ be any such extension.
Since $Y\in\mathcal Q_W$, for any $x,y\in Y$, there exist directed paths from $x$ to $y$ and from $y$ to $x$ in $H_W$ on $Y$.
By the definition of $\mathcal T$, all arcs on these paths belong to $\mathcal T$ and hence to $F'$. 
Since $\pi_F$ is a linear extension, we have 
$\pi_{F'}(x)\leq\pi_{F'}(y)$ and $\pi_{F'}(y)\leq\pi_{F'}(x)$, whcih implies that 
\begin{align*}
\pi_{F'}(x)=\pi_{F'}(y).
\end{align*}
Thus, every parsimonious extension of $(V,F')$ ties all candidates in $Y$.
This construction changes only arcs whose endpoints both belong to $Y$.
Applying the construction successively to all tie-dominant components therefore yields an optimal solution whose every parsimonious extension ties all candidates within every tie-dominant component.
This proves the first assertion.

Finally, let $u$ and $v$ be distinct candidates in the same tie-dominant component $Y$, and suppose that
\begin{align*}
w^\sim_{uv}
>
\max\{w^\succ_{uv},w^\succ_{vu}\}.
\end{align*}
Let $F$ be any optimal solution to $\texttt{WUCF}(H_W,W)$.
Apply the construction above to $Y$ and let $F'=F\cup\mathcal T$.
By the definition of $B_W$, both $(u,v)$ and $(v,u)$ belong to $\mathcal T$. 
Suppose, to the contrary, that $F$ does not contain both arcs.
Then, by \textnormal{(F3)}, $F$ contains exactly one of $(u,v)$ and $(v,u)$.
Thus, $F$ ranks one candidate strictly above the other, whereas $F'$ ties them.
The strict inequality on the weights implies that the contribution of $\{u,v\}$ strictly increases from $F$ to $F'$.
The preceding comparison shows that the contribution of every other pair does not decrease.
Therefore, $g_W(F')>g_W(F)$. 
This contradicts the optimality of $F$.
Since $F$ was arbitrary, every optimal solution to $\texttt{WUCF}(H_W,W)$ contains both arcs $(u,v)$ and $(v,u)$.
\end{proof}

\section{Related work}\label{RelatedWork}

The \texttt{LOP} has been studied extensively from an optimization perspective.
Charon and Hudry~\cite{CharonHudry2007} survey the \texttt{LOP} on weighted and unweighted tournaments.
Grötschel et al.~\cite{Groetscheletal1984} developed a cutting-plane method, while Sukegawa et al.~\cite{Sukegawaetal2011} proposed a Lagrangian heuristics. 
Martí and Reinelt~\cite{MartiReinelt2011} survey exact and heuristic methods. 
These studies allow arbitrary pairwise weights but restrict feasible solutions to strict rankings.
The \texttt{LOPT} retains arbitrary pairwise weights and additionally allows ties.

The Kemeny rule has also been studied from an algorithmic perspective.
Conitzer et al.~\cite{Conitzeretal2006} developed lower bounds on the optimal Kemeny score for exact algorithms.
Ailon et al.~\cite{Ailonetal2008} proposed approximation algorithms for computing Kemeny rankings via feedback arc sets on tournaments.
Kenyon-Mathieu and Schudy~\cite{KenyonMathieuSchudy2007} subsequently developed a polynomial-time approximation scheme.
For rankings with ties, Emond and Mason~\cite{EmondMason2002} developed an exact method based on pairwise comparisons.
Yoo and Escobedo~\cite{YooEscobedo2021} proposed a binary programming formulation that also permits incomplete input rankings.
These studies derive their objective weights from input rankings.
In contrast, the \texttt{LOPT} permits arbitrary weights for the three possible relations associated with each candidate pair.

Finally, Condorcet-type criteria can be used to decompose rank aggregation problems, as discussed in the introduction.
Akbari and Escobedo~\cite{AkbariEscobedo2023} introduced the \texttt{GXCC} for rank aggregation under a parameterized distance with ties.
Their objective weights are determined by the input rankings and the distance parameter.
Our results apply to arbitrary \texttt{LOPT} weights, which need not arise from a distance-based rank aggregation problem.
In addition to fixing strict comparisons between strongly connected components of $H_W$, our results give conditions for ties within a component.

\section{Concluding remarks}\label{Conclusion}

We introduced the linear ordering problem with ties (\texttt{LOPT}).
We extended the properties established for the \texttt{LOP} by Ando et al.~\cite{Andoetal2022} to the \texttt{LOPT}.
This extension also strengthens the results of Yoo and Escobedo~\cite{YooEscobedo2021} for the Kemeny rule with ties.
We further established new properties that determine when candidates can or must be tied in an optimal solution.
These results give a more detailed description of the structure of optimal \texttt{LOPT} solutions.

Future work includes studying the social-choice meaning of this structure, such as its relation to choice sets defined from pairwise comparisons.
Another direction is to develop algorithms that use this structure to solve \texttt{LOPT} instances.

\section*{Acknowledgments}

This work was supported by JSPS KAKENHI Grant Number 25K08185.

\setlength{\bibsep}{0pt}
\bibliographystyle{cas-model2-names}
\bibliography{cas-refs}

\appendix

\section{Proofs for the graph representation}
\label{app:graph-representation}

\begin{proof}[Proof of Observation~\ref{lem:parsimonious-extension}]
Let $G=(V,A)$ be a unicycle-free directed graph.
Let $V_1,V_2,\ldots,V_m$ be the strongly connected components of $G$, indexed in topological order.
Define a ranking with ties $\pi:V\to[m]$ by setting $\pi(x)=i$ for every $x\in V_i$. 

We first show that $\pi$ is a linear extension of $G$. 
Let $(x,y)\in A$.
If $x$ and $y$ belong to the same component, then $\pi(x)=\pi(y)$.
If they belong to distinct components, say, $x\in V_i$ and $y \in V_j$, then we have $i<j$ from $(x,y)\in A$ and hence $\pi(x)<\pi(y)$.
Therefore, $\pi$ is a linear extension of $G$.

We next show that $\pi$ is a parsimonious extension of $G$. 
Suppose that $(x,y)\in A$ and $(y,x)\notin A$. 
Assume that $x$ and $y$ belong to the same strongly connected component.
Then, there is a directed path from $y$ to $x$.
The path together with $(x,y)$ contains a directed cycle that contains $(x,y)$.
However, the reverse of this cycle is not contained in $G$ because $(y,x)\notin A$.
Hence, the cycle is a unicycle.
This contradicts the unicycle-freeness of $G$.
Therefore, $x$ and $y$ belong to distinct components. 
In this case, $\pi(x)<\pi(y)$.
Thus, $\pi$ is parsimonious.

Conversely, let $G$ be a directed graph that admits a parsimonious extension $\pi$. 
Suppose, to the contrary, that $G$ contains a unicycle
$C=(v_1,v_2,\ldots,v_k,v_1)$.
Since $\pi$ is a linear extension, we have
\begin{align*}
\pi(v_1)
\leq
\pi(v_2)
\leq
\cdots
\leq
\pi(v_k)
\leq
\pi(v_1).
\end{align*}
Here, all the inequalities are equalities.
Since the reverse of $C$ is not contained in $G$, we have $(v_{i+1},v_i)\notin A$ for some $i\in[k]$, where $v_{k+1}=v_1$. 
However, since $\pi$ is a parsimonious extension, $\pi(v_i)<\pi(v_{i+1})$. 
This is a contradiction.
Therefore, $G$ is unicycle-free.
\end{proof}
\end{document}